\RequirePackage{fix-cm}

\documentclass[smallextended]{svjour3}      

\smartqed  

\usepackage{graphicx}

\usepackage{epsfig}
\usepackage{amsmath,amssymb,amsfonts}
\usepackage[ruled,vlined,linesnumbered,commentsnumbered]{algorithm2e}

\journalname{Algorithmica}

\begin{document}

\title{On the Tree Augmentation Problem\thanks{A preliminary version appeared in  ESA 2017: 61:1-61:14}}

\author{Zeev Nutov}    

\institute{Z. Nutov \at
              Department of Mathematics and Computer Science \\
							The Open University of Israel, Raanana, Israel. \\
              \email{nutov@openu.ac.il}}

\date{Received: date / Accepted: date}

\maketitle

\newcommand {\ignore} [1] {}

\newcommand{\pa}      {{\sf pa}}
\newcommand{\ch}      {{\sf ch}}
\newcommand{\diam}  {{\sf diam}}
\newcommand{\T}        {{\sf T}}

\newcommand{\TA}    {{\sc Tree Augmentation}}
\newcommand{\CLP}  {{\sc Cut-LP}}
\newcommand{\kBLP} {{\sc $k$-Branch-LP}}
\newcommand{\BLP} {{\sc Bunch-LP}}

\newcommand{\BB}    {{\cal B}}
\newcommand{\LL}    {{\cal L}}
\newcommand{\CC}    {{\cal C}}
\newcommand{\de}    {\delta}
\newcommand{\eps}   {\epsilon}
\newcommand{\la}    {\lambda}
\newcommand{\ga}    {\gamma}
\newcommand{\si}    {\sigma}

\newcommand{\subs}  {\subseteq}
\newcommand{\sem}   {\setminus}
\newcommand{\empt} {\emptyset}

\newcommand{\di}    {\displaystyle}
\newcommand{\ti}    {\tilde}
\newcommand{\h}    {\hat}

\begin{abstract} 
In the {\TA} problem we are given a tree $T=(V,F)$ and a set $E \subs V \times V$ of edges
with positive integer costs $\{c_e:e \in E\}$.
The goal is to augment $T$ by a minimum cost edge set $J \subs E$ such that $T \cup J$ is $2$-edge-connected. 
We obtain the following results.
\begin{itemize}
\item
Recently, Adjiashvili [SODA 17] introduced a novel LP for the problem and used it to break
the $2$-approximation barrier for instances when the maximum cost $M$ of an edge in $E$ is bounded by a constant;
his algorithm computes a $1.96418+\eps$ approximate solution in time $n^{{(M/\eps^2)}^{O(1)}}$.
Using a simpler LP, we achieve ratio $\frac{12}{7}+\eps$ in time $2^{O(M/\eps^2)} poly(n)$.
This gives ratio better than $2$ for logarithmic costs, and not only for constant costs.
% We also show that (for arbitrary costs) the problem admits ratio $3/2$ for trees of diameter $\leq 5$.
\item
One of the oldest open questions for the problem is whether for unit costs (when $M=1$) 
the standard LP-relaxation, so called {\CLP}, has integrality gap less than $2$.
We resolve this open question by proving that for unit costs
the integrality gap of the {\CLP} is at most $28/15=2-2/15$.
In addition, we will prove that another natural LP-relaxation, that is much simpler 
than the ones in previous work, has integrality gap at most $7/4$.
\end{itemize}
\end{abstract}

% \category{F.2.2}{Nonnumerical Algorithms and Problems}{Computations on discrete structures}
% \category{G.2.2}{Discrete Mathematics}{Graph Algorithms}

%%%%%%%%%%%%%%%%%%%%%%
\section{Introduction} \label{s:intro}
%%%%%%%%%%%%%%%%%%%%%%

We consider the following problem:

\begin{center} \fbox{\begin{minipage}{0.965\textwidth} \noindent
{\TA} \\
{\em Input:}  \  \ A tree $T=(V,F)$ and an additional set $E \subs V \times V$ of edges
with positive integer costs $c=\{c_e:e \in E\}$. \\
{\em Output:}   A minimum cost edge set $J \subs E$ such that $T \cup J$ is $2$-edge-connected.
\end{minipage}}\end{center}

The problem was studied extensively, c.f. \cite{FJ,KT,CJR,N,EFKN-APPROX,EFKN-TALG,CKKK,MN,CN,KN-TAP,CG,KN16}.
For a long time the best known ratio for the problem was $2$ for arbitrary costs \cite{FJ} 
and $1.5$ for unit costs \cite{EFKN-APPROX,KN-TAP}; 
see also \cite{EFKN-TALG} for a simple $1.8$-approximation algorithm.
It is also known that the integrality gap of a standard LP-relaxation for the problem, 
so called {\CLP}, is at most $2$ \cite{FJ} and at least $1.5$ \cite{CKKK}.
Several other LP and SDP relaxations were introduced to show that the algorithm in 
\cite{EFKN-APPROX,EFKN-TALG,KN-TAP} achieves ratio better than $2$ w.r.t. to these relaxations, c.f. \cite{CG,KN16}.
For additional algorithms with ratio better than $2$ for restricted versions see \cite{CN,MN}.

Let $M$ denote the maximum cost of an edge in $E$. 
Recently Adjiashvili \cite{A} introduced a novel LP for the problem -- so called the {\sc $k$-Bundle-LP},
and used it to break the natural $2$-approximation barrier for instances when $M$ is bounded 
by a constant. To introduce this result we need some definitions.

The edges of $T$ will be called {\bf $T$-edges} to distinguish them from the edges in $E$.
{\TA} can be formulated as a problem of covering the $T$-edges by paths.
Let $T_{uv}$ denote the unique $uv$-path in $T$.
We say that an {\bf edge $uv$ covers a $T$-edge $f$} if $f \in T_{uv}$. 
Then $T \cup J$ is $2$-edge-connected if and only if $J$ covers $T$.
For a set $B \subs F$ of $T$-edges let $\psi(B)$ denote the set of edges in $E$ that cover some $f \in B$,
and $\tau(B)$ the minimum cost of an edge set in $E$ that covers $B$.
For $J \subs E$ let $x(J)=\sum_{e \in J} x_e$.
The standard LP for the problem which we call the {\CLP} seeks to minimize 
$c^\T x=\sum_{e \in E} c_ex_e$ over the {\sc Cut-Polyhedron} 
$$
\Pi^{Cut}=\left\{x \in \mathbb{R}^E:x(\psi(f))\geq 1  \ \forall f \in F,x \geq 0\right\}    
$$
The {\sc $k$-Bundle-LP} of \cite{A} adds over the standard {\sc Cut-LP} the constraints  
$\sum_{e \in \psi(B)} c_ex_e \geq \tau(B)$ for any forest $B$ in $T$ that has at most $k$ leaves, 
where $k=\Theta(M/\eps^2)$.
The algorithm of \cite{A} computes a $1.96418+\eps$ approximate solution 
w.r.t. the {\sc $k$-Bundle-LP} in time $n^{k^{O(1)}}$.
For unit costs, a modification of the algorithm achieves ratio $5/3+\eps$.

Here we observe that it is sufficient to consider just certain subtrees of $T$ instead of forests.
Root $T$ at some node $r$. The choice of $r$ defines an ancestor/descendant relation on $V$. 
% This defines a partial order on $V$,
% where $u$ is a {\bf descendant} of $v$ and $v$ is an {\bf ancestor} of $u$ if $v$ belongs to $T_{ru}$.
The {\bf leaves of $T$} are the nodes in $V \sem \{r\}$ that have no descendants.
For any subtree $S$ of $T$, the node $s$ of $S$ closest to $r$ is the root of $S$,
and the pair $S,s$ is called a {\bf rooted subtree} of $T,r$;
we will not mention the roots of trees if they are clear from the context. 
We say that $S$ is
a {\bf complete rooted subtree} if it contains all descendants of $s$ in $T$, and 
a {\bf full rooted subtree} if for any non-leaf node $v$ of $S$ the children of $v$ in $S$ and $T$ coincide; 
see Fig.~\ref{f:branch}(a,b).
A {\bf branch of $S$}, or a {\bf branch hanging on $s$}, is a rooted subtree $B$ of $S$ 
induced by the root $s$ of $S$ and the descendants in $S$ of some child $s'$ of $s$; 
see Fig.~\ref{f:branch}~(c).
We say that a subtree $B$ of $T$ is a {\bf branch} if it is a branch of a full rooted subtree, 
or if it is a full rooted subtree with root $r$.
Equivalently, a branch is a union of a full rooted subtree and its parent $T$-edge.

\begin{figure} \centering
\includegraphics{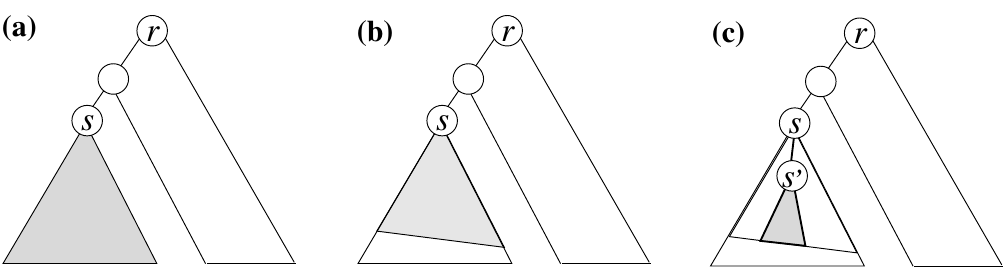}
\caption{(a) complete rooted subtree; (b) full rooted subtree; (c) branch of a full rooted subtree.}
\label{f:branch}
\end{figure}

Let $\BB_k$ denote the set of branches in $T$ with less than $k$ leaves. 
The {\kBLP} seeks to minimize $c^\T x=\sum_{e \in E} c_ex_e$ over the {\sc $k$-Branch-Polyhedron} 
$\Pi^{Br}_k \subs \mathbb{R}^E$ defined by the constraints:
\begin{eqnarray*} 
\sum_{e \in \psi(f)} x_e        & \geq & 1  \ \ \ \ \ \ \ \ \ \ \forall f \in F             \\ 
\sum_{e \in \psi(B)} c_ex_e  & \geq & \tau(B) \ \ \ \ \     \forall B \in \BB_k   \\
x_e                                          & \geq & 0 \ \ \ \ \ \ \ \ \ \ \forall e \in E       
\end{eqnarray*}

The set of constrains of the {\sc $k$-Branch-LP} is a subset of constraints of the
 {\sc $k$-Bundle-LP} of \cite{A}, hence the {\sc $k$-Branch-LP} is 
both more compact and its optimal value is no larger than that of the {\sc $k$-Bundle-LP}.
The first main result in this paper is: 

\begin{theorem} \label{t:main}
For any $1 \leq \lambda \leq k-1$, {\TA} admits a $4^k \cdot poly(n)$ time algorithm that 
computes a solution of cost at most 
$\rho+\frac{8}{3}\frac{\la M}{k-\la M}+\frac{2}{\la}$ times the optimal value of the {\kBLP}, 
where $\rho=\frac{12}{7}$ for arbitrary costs and $\rho=1.6$ for unit costs.
\end{theorem}

For a given $\eps$, choosing properly $\la=\Theta(1/\eps)$ and $k=\Theta(M/\eps^2)$ 
gives ratio $\rho+\eps$ in time $2^{O(M/\eps^2)} \cdot poly(n)$.

In parallel to our work Fiorini, Gro\ss, K\"{o}nemann, and Sanit\'{a} \cite{FGKS} 
augmented the {\sc $k$-Bundle LP} of \cite{A} by additional constraints -- $\{0,\frac{1}{2}\}$-Chv\'{a}tal-Gomory Cuts,
to achieve ratio $1.5+\eps$ in $n^{{(M/\eps^2)}^{O(1)}}$ time, 
thus almost matching the best known ratio for unit costs \cite{EFKN-APPROX,KN-TAP}.  
Our result in Theorem~\ref{t:main}, done independently, shows 
that already the {\sc $k$-Bundle LP} has integrality gap closer to $1.5$ than to $2$.
% We also note that the idea to use $\{0,\frac{1}{2}\}$-Chv\'{a}tal-Gomory Cuts
% was presented by the author earlier in Dagstuhl 2016.
Our version of the algorithm of \cite{A} is also simpler than the one in \cite{FGKS}.
In fact, combining our approach with \cite{FGKS} enables to achieve ratio 
$1.5+\eps$ in $2^{O(M/\eps^2)} \cdot poly(n)$ time. 
Note that this allows to achieve this ratio for logarithmic costs, and not only for constant costs.
We will provide an additional comparison of our results and those in \cite{FGKS} in Section~\ref{ss:FGKS}.

Very recently  the natural ratio $1.5$ was improved to a smaller constant by Grandoni, Kalaitzis \& Zenklusen \cite{GCZ} . 
Their approach also works for small integer costs and gives ratio that tends to $1.5$ from below. 
% We do not know whether one can get a better than 1.5 approximation also for logarithmic costs using the approach in this paper.

We note that while the running time of the combinatorial algorithm for unit costs of \cite{KN-TAP} is roughly the same as that of finding a maximum matching,
the recent algorithms \cite{FGKS,GCZ} that are based on the approach of Adjiashvili \cite{A} have running time 
$n^{{(M/\eps^2)}^{O(1)}}$, with large constant hidden on the $O(\cdot)$ term; 
this is very high even for unit costs and $\eps=0.1$. 
Our result in Theorem~\ref{t:main} substantially reduces the running time to $2^{O(M/\eps^2)} \cdot poly(n)$.
% with small constant hidden in the $O(\cdot)$ term.

Let $\diam(T)$ denote the diameter of $T$.
{\TA} admits a polynomial time algorithm when $\diam(T) \leq 3$. 
If $\diam(T)=2$ then $T$ is a star and we get the {\sc Edge-Cover} problem, while the case $\diam(T)=3$ 
is reduced to the case $\diam(T)=2$ by ``guessing'' some optimal solution edge that covers the central $T$-edge.   
The problem becomes NP-hard when $\diam(T)=4$ even for unit costs \cite{FJ}.
We prove that (without solving any LP)  for arbitrary costs
{\TA} with trees of diameter $\leq 5$ admits ratio $3/2$.

Our second main result resolves one of the oldest open questions concerning the problem --
whether for unit costs the integrality gap of the {\CLP} is less than $2$.
This was conjectured in the 90's by Cheriyan, Jord\'{a}n \& Ravi \cite{CJR} for arbitrary costs,
but so far there was no real evidence for this even for unit costs. 
Our second main result resolves this old open question.

\begin{theorem} \label{t:gap}
For unit costs, the integrality gap of the {\CLP}  is at most $28/15=2-2/15$.
\end{theorem}

In addition, we will show that for unit costs, another natural simple LP-relaxation, 
has integrality gap at most $7/4$.

% In our opinion, the main open problem for {\TA} is obtaining a ratio better than $2$ for arbitrary costs, if this is possible.
% This paper makes a progress in this direction by showing ratio better than $2$ for logarithmic costs,

%%%%%%%%%%%%%%%%%%%%%%%%%%%%
\section{Algorithm for bounded costs (Theorem~\ref{t:main})}
%%%%%%%%%%%%%%%%%%%%%%%%%%%%

The Theorem~\ref{t:main} algorithm is a modification of the algorithm of \cite{A}. 
We emphasize some differences. We use the {\kBLP} instead of the {\sc $k$-Bundle-LP} of \cite{A}. 
But, unlike \cite{A}, we do not solve our LP at the beginning.
Instead, we combine binary search with the ellipsoid algorithm as follows.
We start with lower and upper bounds $p$ and $q$ on the value of the {\kBLP}, e.g., 
$p=0$ and $q$ is the cost of some feasible solution to the problem.
Given a ``candidate'' $x$ with $q \leq c^\T x \leq p$,
the outer iteration (see Algorithm~\ref{alg:rounding}) of the entire algorithm 
either returns a solution of cost at most $(\rho+\frac{8}{3}\frac{\la M}{k-\la M}+\frac{2}{\la})c^\T x$ 
or a constraint of the {\kBLP} violated by $x$;
we show that this can be done in time $4^k \cdot poly(n)$, rather than in time $n^{k^{O(1)}}$ as in \cite{A}.
We set $p \gets \frac{p+q}{2}$ in the former case and $q \gets \frac{p+q}{2}$ in the latter case 
and continue to the next iteration, terminating when $p-q$ is small enough.
This essentially gives a $4^k \cdot poly(n)$ time separation oracle for the {\kBLP} 
(if a violated $k$-branch constraint is found).
Since the ellipsoid algorithm uses a polynomial number of calls to the separation oracle,
the running time is $4^k \cdot poly(n)$.
Note that checking whether $x \in \Pi^{Cut}$ is trivial, 
hence for simplicity of exposition we will assume that the ``candidate'' $x$ is in $\Pi^{Cut}$.

For a set $S$ of $T$-edges we denote by $T/S$ the tree obtained from $T$ by contracting every $T$-edge of $S$.
This defines a new {\TA} instance (that may have loops and parallel edges),
where contraction of a $T$-edge $uv$ leads to shrinking $u,v$ into a single node in the graph $(V,E)$ of edges. 
In the algorithm, we repeatedly take a certain complete rooted subtree $\hat{S}$, and either find 
a $k$-branch-constraint violated by some branch in $\hat S$, or a ``cheap'' cover $J_S$ of a subset $S$ 
of the $T$-edges of $\hat S$; in the latter case, we add $J_S$ to our partial solution $J$, 
contract $\hat S$, and iterate on the instance $T \gets T/\hat S$.
At the end of the loop, the edges that are still not covered by the partial solution $J$ are covered by a different 
procedure, by a total cost $\frac{2}{\la} \cdot c^\T x$, as follows.

We call a $T$-edge $f \in F$ {\bf $\la$-thin} if $x(\psi(f)) \leq \la$, and $f$ is {\bf $\la$-thick} otherwise.
We need the following lemma from \cite{A}, for which we provide a proof for completeness of exposition.

\begin{lemma} [\cite{A}] \label{l:thick}
There exists a polynomial time algorithm that given $x \in \Pi^{Cut}$, $\la>1$, and a set $F' \subs F$ of 
$\la$-thick $T$-edges computes a cover $J'$ of $F'$ of cost $\leq \frac{2}{\la} \cdot c^\T x$. 
\end{lemma}
\begin{proof}
Since all $T$-edges in $F'$ are $\la$-thick, $x/\la$ is a feasible solution to the {\CLP} for covering $F'$.
Thus any polynomial time algorithm that computes a solution $J'$ of cost at most 
$2$ times the optimal value of the {\CLP} for covering $F'$ has the desired property.
There are several such algorithms, see \cite{FJ,GGPS,Jain}.
\qed
\end{proof}

% Note that if $f$ is $\la$-thin then $\sum_{e \in \psi(f)} \leq \la M$.
We say that a complete rooted subtree $S$ of $T$ is a {\bf $(k,\la)$-subtree}
if $S$ has at least $k$ leaves and if either the parent $T$-edge $f$ of $S$ is $\la$-thin or $s=r$.
For $\la=\Theta(1/\eps)$ and  $k=\Theta(M/\eps^2)$ we choose $\hat S$ to be 
{\em an inclusionwise minimal $(k,\la)$-subtree}.
Let us focus on the problem of covering such $\hat S$.
Let $S'$ be the set of $T$-edges of the inclusionwise maximal subtree of $\hat S$ that contains the root $s$ of $\hat S$ 
and has only $\la$-thick $T$-edges (possibly $S=\empt$); see Fig.~\ref{f:k-trees}(a).
We postpone covering the $T$-edges in $S'$ to the end of the algorithm,
so we contract $S'$ into $s$ and consider the tree $S \gets \hat S/S'$; see Fig.~\ref{f:k-trees}(b). 
In $S$, every branch $B$ hanging on $s$ has less than $k$ leaves, by the minimality of $S$,
hence it has a corresponding constraint in the {\kBLP}. 
We will show that for a $k$-branch $B$ an optimal set of edges that covers $B$ 
can be computed in time $4^k \cdot poly(n)$.
If $\sum_{e \in \psi(B)} c_ex_e<\tau(B)$ for some branch $B$ hanging on $s$ in $S$, 
then we return the corresponding $k$-branch constraint violated by $x$; 
otherwise, we will show how to compute a ``cheap'' cover of $S$. 
More formally, in the next section we will prove:

\begin{figure} 
\centering
\includegraphics[width=9cm]{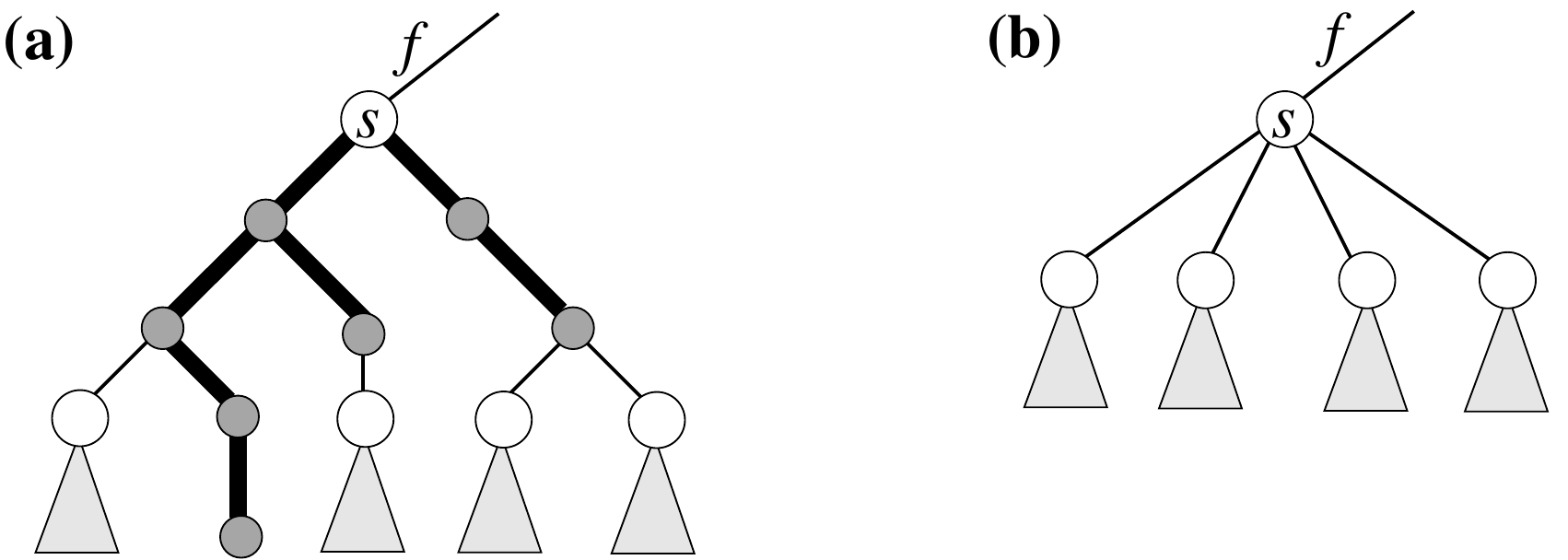}
\caption{Branches hanging on $s$ after contracting $S'$; $\la$-thick $T$-edges are shown by thick lines.}
\label{f:k-trees}
\end{figure}

\begin{lemma} \label{l:main}
Suppose that we are given a {\TA} instance and $x \in \Pi^{Cut}$ 
such that any complete rooted proper subtree of the input tree has less than $k$ leaves.
Then there exists a $4^k \cdot poly(n)$ time algorithm that either finds 
a $k$-branch constraint violated by $x$, or computes a solution 
of cost $\leq \rho \sum_{e \in E \sem R}c_ex_e+\frac{4}{3} \sum_{e \in R} c_ex_e$,
where $\rho$ is as in Theorem~\ref{t:main} and $R$ is the set of edges in $E$ incident to the root.
\end{lemma}

To find a cheap cover of $S$, we consider the {\TA} instance obtained from $T/S'$ 
by contacting  into $s$ all nodes not in $S$. 
Note that every edge that was in $\psi(S) \cap \psi(f)$ is now incident to the root.
Thus since $\rho \geq \frac{4}{3}$, Lemma~\ref{l:main} implies: 

\begin{corollary} \label{c:main}
There exists a $4^k \cdot poly(n)$ time algorithm that either finds 
a $k$-branch-constraint violated by $x$, or a cover $J_S$ of $S$
of cost \ $c(J_S) \leq \rho \sum_{e \in \ga(S)}c_ex_e+\frac{4}{3}\sum_{e \in \psi(f)}c_ex_e$,
where $\rho$ is as in Theorem~\ref{t:main} and $\ga(S)$ denotes the set of edges with both endnodes in $S$,
and $f$ is the parent $T$-edge of $S$.
\end{corollary}

The outer iteration of the algorithm is as follows:

\medskip \medskip

\begin{algorithm}[H]
\caption{{\sc Outer-Iteration}$(T=(V,F),E,x,c,k,r,\la)$}
\label{alg:rounding}
$J \gets \empt$, $F' \gets \empt$ \\
\While{\em $T$ has at least $2$ nodes}
{
let $\hat S$ be an inclusionwise minimal $(k,\la)$-subtree of $T$ \\
let $S'$ be the edge-set of the inclusionwise maximal subtree of $\hat S$ that contains the root $s$ of $\hat S$
and has only $\la$-thick edges \\
apply the algorithm from {\bf Corollary~\ref{c:main}} on $S \gets \hat S/S'$ \\
if {\bf Corollary~\ref{c:main}} algorithm returns  a cover $J_S$ of $S$ then do: \newline
$F' \gets F' \cup S'$, $J \gets J \cup J_S$, $T \gets T/\hat S$ \\
else, return a $k$-branch constraint violated by $x$ and STOP 
}
compute a cover $J'$ of $F'$ of cost $c(J') \leq \frac{2}{\la} \cdot c^\T x$ using {\bf Lemma~\ref{l:thick}} algorithm \\
\Return{$J \cup J'$}
\end{algorithm}

\medskip \medskip

Note that at step~7 the $T$-edges in $F'$ are all $\la$-thick and thus Lemma~\ref{l:thick} applies. 
We will now analyze the performance of the algorithm assuming than 
no $k$-branch-constraint violated by $x$ was found. 
Let $\de(S)$ denote the set of edges with exactly one endnode in $S$ 
and $\ga(S)$ the set of edges with both endnodes in $S$.
Let $f$ be the parent $T$-edge of $S$.
Since $f$ is $\la$-thin 
$$
\sum_{e \in \psi(f)}c_ex_e \leq \sum_{e \in \psi(f)}Mx_e \leq M \cdot x(\psi(f)) \leq M\la \ .
$$
Since $x(\de(v)) \geq 1$ for every leaf $v$ of $S$, $c_e \geq 1$ for every $e \in E$, 
and since $S$ is a $(k,\la)$-subtree 
$$
2\sum_{e \in \ga(\hat S)} c_ex_e = \sum_{v \in \hat S} \sum_{e \in \de(v)}c_ex_e-\sum_{e \in \psi(f)}c_ex_e 
\geq \sum_{v \in \hat S \sem \{s\}} x(\de(v))-\la M  \geq k-\la M \ .
$$

Consider a single iteration in the while-loop. Let $\Delta(c^\T x)$ denote 
the decrease in the LP-solution value as a result of contracting $\hat S$. Then
$$
\Delta(c^\T x)=\sum_{e \in \ga(\hat S)} c_ex_e \geq \frac{k-\la M}{2} \ .
$$
On the other hand, by Lemma~\ref{l:main}, the partial solution cost increases by at most
$$
c(J_S) \leq \rho \sum_{e \in \ga(S)}c_ex_e + \frac{4}{3}\sum_{e \in \psi(f)} c_e x_e 
         \leq \rho \sum_{e \in \ga(\hat S)}c_ex_e + \frac{4}{3}\la M \ .
$$
Thus
$$
\frac{c(J_S)}{\Delta(c^\T x)} \leq \rho+\frac{8}{3} \frac{\la M}{k-\la M} \ .
$$
The while-loop terminates when the LP-solution value becomes $0$, hence by a standard local-ratio/induction 
argument we get that at the end of the while-loop $c(J) \leq \left(\rho+\frac{8}{3}\frac{\la M}{k-\la M}\right)c^\T x$.
At step~7 we add an edge set of cost $\leq \frac{2}{\la}c^\T x$, and Theorem~\ref{t:main} follows.
It only remains to prove Lemma~\ref{l:main}, which we will do in the subsequent sections.

%%%%%%%%%%%%%%%%%%%%%%%%%
\subsection{Proof of Lemma~\ref{l:main}} \label{ss:main}
%%%%%%%%%%%%%%%%%%%%%%%%%

Assume that we are given an instance $T=(V,F),E,c$ of {\TA} with root $r$ and $x$ as in Lemma~\ref{l:main}. 
It is known that {\TA} instances when $T$ is a path can be solved in polynomial time. 
Thus by a standard ``metric completion'' type argument 
we may assume that the graph $(V,E)$ is a complete graph and that $c_{uv}=\tau(T_{uv})$ for all $u,v \in V$.
Indeed, then for each $u,v \in V$ corresponds a set $P_{uv}$ of edges of cost $c_{uv}=\tau(T_{uv})$ that covers $T_{uv}$,
and whenever and edge $uv$ is chosen to the solution, it can be replaced by $P_{uv}$ without increasing the cost.
Note that we use this assumption only in the proof of Lemma~\ref{l:main},
where the running time does not depend on the maximum cost $M$ of an edge in $E$. 
Let us say that an edge $uv \in E$ is: 
\begin{itemize}
\item
a {\bf cross-edge} if $r$ is an internal node of $T_{uv}$;
\item
an {\bf in-edge} if $r$ does not belong to $T_{uv}$;
\item
an {\bf $r$-edge} if $r=u$ or $r=v$;
\item
an {\bf up-edge} if one of $u,v$ is an ancestor of the other.
\end{itemize}

For a subset $E' \subs E$ of edges the {\bf $E'$-up vector of $x$} is obtained from $x$ as follows:
for every  non-up edge $e=uv \in E'$ increase $x_{ua}$ and $x_{va}$ by $x_e$ and then reset $x_e$ to $0$,
where $a$ is the least common ancestor of $u$ and $v$.
The {\bf fractional cost} of a set $J$ of edges w.r.t. $c$ and $x$ is defined by $\sum_{e \in J} c_ex_e$.
Let $C_x^{\sf in}$, $C_x^{\sf cr}$, and $C_x^r$ denote the fractional cost of in-edges, 
cross-edges, and $r$-edges, respectively, w.r.t. $c$ and $x$.
We fix some $x^* \in \Pi^{Cut}$ and denote by $C^{\sf in}$, $C^{\sf cr}$, and $C^r$ 
the fractional cost of in-edges, cross-edges, and $r$-edges, respectively, w.r.t. $c$ and $x^*$.
We give two rounding procedures, given in Lemmas \ref{l:round1} and \ref{l:round2}.
The rounding procedure in Lemma~\ref{l:round1} is similar to that of \cite{A},
but we show that it can be implemented in time $4^k \cdot poly(n)$ instead of $n^{k^{O(1)}}$.
% for which we present an improved running time analysis.

\begin{lemma} \label{l:round1}
There exists a $4^k \cdot poly(n)$ time algorithm that 
either finds a $k$-branch inequality violated by $x^*$, 
or returns an integral solution of cost at most $C^{\sf in}+2C^{\sf cr}+C^r$.
\end{lemma}
\begin{proof}
Let $\BB$ be the set of branches hanging on $r$.
For every $B \in \BB$ compute an optimal solution $J_B$.
If for some $B \in \BB$ we have $\tau(B)>\sum_{e \in \psi(B)}c_ex^*_e$ then a $k$-branch inequality violated by $x^*$ is found.
Else, the algorithm returns the union $J=\bigcup_{B \in \BB} J_B$ of the computed edge sets. 
As every cross-edge has its endnodes in two distinct branches, while every 
in-edge or $r$-edge has its both endnodes in the same branch, we get
\begin{eqnarray*}
c(J) & \leq & \sum_{B \in \BB} \tau(B) \leq \sum_{B \in \BB} \sum_{e \in \psi(B)}c_ex^*_e
= \sum_{B \in \BB} \left(\sum_{e \in \de(B)}c_ex^*_e+\sum_{e \in \ga(B)}c_ex^*_e\right) \\
      &   =    & 2C^{\sf cr}+C^{\sf in}+C^r \ .
\end{eqnarray*}

It remains to show that an optimal solution in each branch of $r$  
can be computed in time $4^k \cdot poly(n)$. More generally, we will show that 
{\TA} instances with $k$ leaves can be solved optimally within this time bound.
Recall that we may assume that the graph $(V,E)$ is a complete graph and that 
$c_{uv}=\tau(T_{uv})$ for all $u,v \in V$.
We claim that then we can assume that $T$ has no node $v$ with $\deg_T(v)=2$.
This is a well known reduction (e.g. see \cite{MV}).
In more details, we show that any solution $J$ can be converted into a solution of no greater cost 
that has no edge incident to $v$, and thus $v$ can be ``shortcut''.
If $J$ has edges $uv,vw$ then it is easy to see that 
$(J \sem \{uv,vw\}) \cup \{uw\}$ is also a feasible solution,
of cost at most $c(J)$, since $c_{uw} \leq c_{uv}+c_{vw}$.  
Applying this operation repeatedly we may assume that $\deg_J(v) \leq 1$.
If $\deg_J(v)=0$, we are done.
Suppose that $J$ has a unique edge $e=vw$ incident to $v$. 
Let $vu$ and $vu'$ be the two $T$-edges incident to $v$, where assume that $vu'$ is not covered by $e$.
Then there is an edge $e' \in J$ that covers $vu'$. Since $e'$ is not incident to $v$,
it must be that $e'$ covers $vu$. Replacing $e$ by the edge $wu$ gives a feasible 
solution without increasing the cost.

Consequently, we reduce our instance to an equivalent instance with at most $2k-1$ tree edges.
Now recall that {\TA} is a particular case of the {\sc Min-Cost Set-Cover} problem,
where the set $F$ of $T$-edges are the elements and $\{T_e:e \in E\}$ are the sets.
The {\sc Min-Cost Set-Cover} problem can be solved in $2^n \cdot poly(n)$ time
via dynamic programming, where $n$ is the number of elements;
such an algorithm is described in \cite[Sect. 6.1]{C-book} for unit costs,
but the proof extends to arbitrary costs \cite{MC}.
Thus our reduced {\TA} instance can be solved in 
$2^{2k-1} \cdot poly(n) \leq 4^k \cdot poly(n)$ time.
\qed
\end{proof}

For the second rounding procedure \cite{A} proved that for any $\la>1$ 
one can compute in polynomial time an integral solution of cost at most
$2\la C^{\sf in}+\frac{4}{3} \frac{\la}{\la-1}C^{\sf cr}$ . We prove:

\begin{lemma} \label{l:round2}
There exists a polynomial time algorithm that computes a solution 
of cost $\frac{4}{3}(2C^{\sf in}+C^{\sf cr}+C^r)$, 
and a solution of size $2C^{\sf in}+\frac{4}{3}C^{\sf cr}+C^r$ in the case of unit costs.
\end{lemma}

Consider the case of arbitrary bounded costs. 
If $C^{\sf in} \geq \frac{2}{5}C^{\sf cr}$ we use the rounding procedure from Lemma~\ref{l:round1}
and the rounding procedure from Lemma~\ref{l:round2} otherwise.
In both cases we get $c(J) \leq  \frac{12}{7}(C^{\sf in}+C^{\sf cr})+\frac{4}{3}C^r$.
In the case of unit costs,
if $C^{\sf in} \geq \frac{2}{3}C^{\sf cr}$ we use the rounding procedure from Lemma~\ref{l:round1},
and the procedure from Lemma~\ref{l:round2} otherwise.
In both cases we get $c(J) \leq 1.6(C^{\sf in}+C^{\sf cr})+C^r$.

Lemma~\ref{l:round2} is proved in the next section.
The proof relies on properties of extreme points of the {\sc Cut-Polyhedron} $\Pi^{Cut}$ 
% given in Lemmas \ref{l:half} and \ref{l:cycles}; 
that are of independent interest.

%%%%%%%%%%%%%%%%%%%%%%%%%%%%%%%%%%%%%%%%%%%%%%%%%%%%%%%%
\subsection{Properties of extreme points of the {\sc Cut-Polyhedron} (Lemma~\ref{l:round2})} \label{ss:round2}
%%%%%%%%%%%%%%%%%%%%%%%%%%%%%%%%%%%%%%%%%%%%%%%%%%%%%%%%

W.l.o.g., we augment the {\CLP} by the constraints $x_e \leq 1$ for all $e \in E$, 
while using the same notation as before.
The (modified) {\CLP} always has an optimal solution $x$
that is an {\bf extreme point} or a {\bf basic feasible solution} of $\Pi^{Cut}$.
Geometrically, this means that $x$ is not a convex combination of other points in $\Pi^{Cut}$; 
algebraically this means that there exists a set of $|E|$ inequalities in the system defining $\Pi^{Cut}$ 
such that $x$ is the unique solution for the corresponding linear equations system.
These definitions are known to be equivalent and we will use both of them, c.f. \cite{LRS}.

A set family $\LL$ is {\bf laminar} if any two sets in the family are either 
disjoint or one contains the other. 
Note that {\TA} is equivalent to the problem of covering the laminar family 
of the node sets of the complete rooted proper subtrees of $T$, where an edge covers a node set $S$
if it has exactly one endnode in $S$.
In particular, note that the constraint $x(\psi(f)) \geq 1$ is equivalent to the 
constraint $x(\de(S)) \geq 1$ where $S$ is the node set of the complete rooted subtree with parent $T$-edge $f$. 
Let $\mathbb{N}_0$ denote the set of non-negative integers.

\begin{lemma} \label{l:half}
Let $(V,E)$ be a graph, $\LL$ a laminar family on $V$, and $b \in \mathbb{N}_0^\LL$.
Suppose that for every $S \in \LL$ there is no edge between two distinct children of $S$ and that
the equation system $\{x(\de(S))=b_S:S \in \LL\}$ has a unique solution $0<x^*<1$.
Then $x^*_e=1/2$ for all $e \in E$. Furthermore, each endnode of every $e \in E$ belongs to some $S \in \LL$.
\end{lemma}
\begin{proof}
For every $uv \in E$ put one token at $u$ and one token at $v$.
The total number of tokens is $2|E|$.
For $S \in \LL$ let $t(S)$ be the number of tokens placed at nodes in $S$ that belong to no child of $S$.
Since $\LL$ is laminar, every token is placed in at most one set in $\LL$, and thus $\sum_{S \in \LL} t(S) \leq 2|E|$. 
Let $S \in \LL$ and let $\CC(S)$ be the set of children of $S$ in $\LL$.
Let $E_S$ be the set of edges in $\de(S)$ that cover no child of $S$,
and $E_{\CC(S)}$ the set of edges not in $\de(S)$ that cover some child of $S$.
% $$E_S=\de(S) \sem \left(\bigcup_{C \in \CC(S)}\de(C) \right) \ \ \ \ \ \ 
% E_{\CC(S)}=\left(\bigcup_{C \in \CC(S)}\de(C) \right) \sem \de(S)$$
Note that no $e \in E_{\CC(S)}$ connects two distinct children of $S$. Observe that
$$
x^*(E_S)- x^*(E_{\CC(S)})=x^*(\de(S))-\sum_{C \in \CC(S)} x^*(\de(C))=b_S-\sum_{C \in \CC(A)} b_C \equiv b'_S \ .
$$
Thus $x^*(E_S)- x^*(E_{\CC(S)})$ is an integer.
We cannot have $|E_S|=|E_{\CC(S)}|=0$ by linear independence, and we cannot have 
$|E_S|+|E_{\CC(S)}|=1$ by the assumption $0<x<1$.
Thus $|E_S|+|E_{\CC(S)}| \geq 2$.
Since no $e \in E$ goes between children of $S$, $t(S) \geq |E_S|+|E_{\CC(S)}|$.
Consequently, since $\sum_{S \in \LL} t(S) \leq 2|E|$, we get:
$t(S) = |E_S|+|E_{\CC(S)}| = 2 \ \forall S \in \LL$.
Moreover, if an endnode of some $e \in E$ belongs to no $S \in \LL$,
then we get the contradiction $\sum_{S \in \LL} t(S) \geq 2|E|+1$.
Now we replace our equation system by an equivalent one $\left\{x(E_S)-x(E_{\CC(S)})=b'_S:S \in \LL\right\}$
obtained by elementary operations on the rows of the coefficients matrix.
Note that $x^*$ is also a unique solution to this new equation system.
Moreover, this equation system has exactly two variables in each equation and all its coefficients are integral.
By \cite{HMNT}, the solution of such equation systems is always half-integral.
% The proof of the lemma is complete.
\qed
\end{proof}

Let us say that {\TA} instance is {\bf spider-shaped} if every in-edge in $E$ is an up-edge.
By a standard ``iterative rounding'' argument (c.f. \cite{LRS}), and using the correspondence between 
rooted trees and laminar families, we get from Lemma~\ref{l:half}:

% \begin{lemma} \label{l:half'}
% Let $(V,E)$ be a graph, $\LL$ a laminar family on $V$, and $b \in \mathbb{N}^\LL_+$.
% Suppose that for every $A \in \LL$ there is no edge between two distinct children of $A$.
% Then all the extreme points of the polytope
% $\Pi^{\LL,b}=\{x \in \mathbb{R}^E:x(\de(A)) \geq b_A \ \forall A \in \LL, 0 \leq x \leq 1\}$ are half-integral,
% namely, have entries in $\{0,\frac{1}{2},1\}$.
% Furthermore, $x_e \in \{0,1\}$ for every $e \in E$ that has an endnode that belongs to no $A \in \LL$.
% \end{lemma}

\begin{corollary} \label{c:half}
Suppose that we are given a spider-shaped {\TA} instance and $b \in \mathbb{N}_0^F$.
Let $x$ be an extreme point of the polytope 
$\{x \in \mathbb{R}^E:x(\psi(f)) \geq b_f \ \forall f \in F, 0 \leq x \leq 1\}$. 
Then $x$ is half-integral (namely, $x_e \in \{0,\frac{1}{2},1\}$ for all $e \in E$)
and $x_e \in \{0,1\}$ for every $e \in \de(r)$.
\end{corollary}
\begin{proof}
Note that the assumptions of the lemma remain valid after each one of the following two operations:
\begin{itemize}
\item
If $x_e=0$ for some $e \in E$ then remove $e$.
\item
If $x_e=1$ for some $e \in E$ then set $b_f \gets \max\{b_f-1,0\}$ for every $f \in T_e$, and remove $e$.
\end{itemize}
We thus can argue by induction. The base case $|E|=1$ is obvious.
If there is $e \in E$ with $x_e \in \{0,1\}$ then apply an appropriate operation above and the induction hypothesis.
Otherwise, $0 < x <1$.
Let $\LL_T$ be the laminar family of $T$.
Since $x$ is an extreme point, there exists a laminar family $\LL \subs \LL_T$
such that $x$ is the unique solution to the equation system $\{x(\de(S))=b_S:S \in \LL\}$.
Note that since our {\TA} instance is spider-shaped, there is no edge that connects two children of some $S \in \LL$. 
Thus Lemma~\ref{l:half} implies $x^*_e=1/2$ for all $e \in E$, and the proof is complete.
\qed
\end{proof}

The algorithm that computes an integral solution 
of cost $\frac{4}{3}(2C^{\sf in}+C^{\sf cr}+C^r)$ is as follows.
% Assume that the graph $(V,E)$ is a complete graph and that 
% $c_{uv}=\tau(T_{uv})$ for all $u,v \in V$,
We obtain a spider-shaped instance by removing all non-up in-edges and 
compute an optimal extreme point solution $x$ to the {\CLP}. 
By Corollary~\ref{c:half}, $x$ is half-integral and $x_e \in \{0,1\}$ for every $e \in \de(r)$.
We take into our solution every edge $e$ with $x_e=1$
and round the remaining $1/2$ entries using the algorithm of
Cheriyan, Jord\'{a}n \& Ravi \cite{CJR}, that showed how to round 
a half-integral solution to the {\sc Cut-LP} to integral solution within a factor of $4/3$.
Thus we can compute a solution $J$ of cost at most 
$c(J) \leq \frac{4}{3}c^\T x \leq \frac{4}{3}c^\T x^*$.
We claim that $c^\T x \leq 2C^{\sf in}+C^{\sf cr}+C^r$.
To see this let $E^{\sf in}$ be the set of in-edges and let $x'$ be the  $E^{\sf in}$-up vector of $x^*$.
Then $x'$ is a feasible  solution to the {\CLP} of value $2C^{\sf in}+C^{\sf cr}+C^r$,
in the obtained {\TA} instance with all non-up in-edges removed.
But since $x$ is an optimal solution to the same LP, we have 
$c^\T x \leq c^\T x' = 2C^{\sf in}+C^{\sf cr}+C^r$.
This concludes the proof of Lemma~\ref{l:round2} for the case of arbitrary costs.

In the rest of this section we consider the case of unit costs. % For this case we prove:

\begin{lemma} \label{l:cycles}
Let $a,b \geq 0$ and let $x$ be an extreme point of the polytope
$$
\Pi=\{x \in \Pi^{Cut}:C_x^{\sf in}=a,C_x^{\sf cr}=b\} 
$$
such that $x_e >0$ for every cross-edge $e$.  
Then the graph $(V,E^{\sf cr})$ of cross-edges has no even cycle 
and each one of its connected components has at most one cycle.
\end{lemma}

\begin{proof}
Let $Q$ be a cycle in $E^{\sf cr}$ and let $\eps=\min_{e \in Q} x_e$.
% Note that if $e_i \in \psi(f)$ for some $f \in F$ then $e_i+1 \in \psi(f)$ or $e_{i-1} \in \psi(f)$, 
% where the indices are modulo $k$. This implies that $x_e <1$ for every $e \in Q$, 
% as otherwise there is no $f \in F$ such that $e \in \psi(f)$ and $x(\psi(f))=1$, 
% contradicting that $x$ is an extreme point.
Since $x_e>0$ for all $e \in E^{\sf cr}$, $\eps>0$. If $|Q|$ is even,  
let $Q',Q''$ be a partition of $Q$ into two perfect matchings. 
Let $z$ be a vector defined by $z_e=\eps$ if $e \in Q'$, $z_e=-\eps$ if $e \in Q''$, and $z_e=0$ otherwise.
By the choice of $\eps$, $x+z,x-z$ are non-negative, and it is not hard to verify that $x+z,x-z \in \Pi$.
However, $x=\frac{1}{2}(x+z)+\frac{1}{2}(x-z)$, contradicting that $x$ is an  extreme point.

Suppose that $|Q|$ is odd.
Let $u,v$ be nodes on $Q$, possibly $u=v$.
We claim that $(V,E^{\sf cr} \sem Q)$ has no $uv$-path;
this also implies that any two odd cycles in $(V,E^{\sf cr})$ are node disjoint.
Suppose to the contrary that $(V,E^{\sf cr} \sem Q)$ has a $uv$-path $P$.
Let $P'$ and $P''$ be the two internally disjoint $uv$-paths in $Q$
where $|P'|$ is odd and $|P''|$ is even.
Then one of $P \cup P'$ and $P \cup P''$ is an even cycle, 
contradicting that $(V,E^{\sf cr})$ has no even cycles.

\begin{figure} 
\centering
\includegraphics[width=12cm]{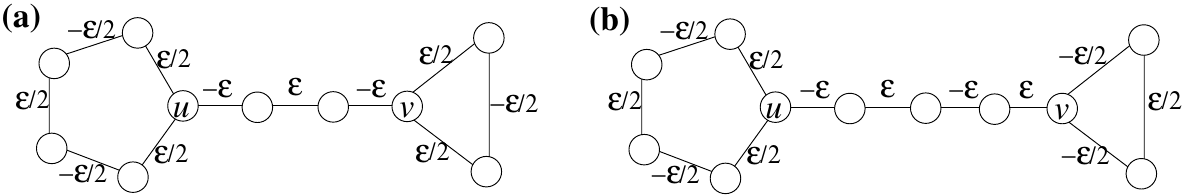}
\caption{Illustration to the proof that no two cycles in $(V,E^{\sf cr})$ are connected by a path.}
\label{f:two-cycles}
\end{figure}

Finally, we show that no two cycles in $(V,E^{\sf cr})$ are connected by a path.
Suppose to the contrary that $(V,E^{\sf cr})$ has a $uv$-path $P$ that connects two distinct cycles 
$Q_u$ and $Q_v$, see Fig.~\ref{f:two-cycles}.
Let $z$ be defined as in Fig.~\ref{f:two-cycles}.
By the choice of $\eps$, each one of the vectors $x+z$ and $x-z$ is non-negative,
and they are both in $\Pi$.
However, $x=\frac{1}{2}(x+z)+\frac{1}{2}(x-z)$, contradicting that $x$ is an  extreme point.
\qed
\end{proof}

Note that Lemma~\ref{l:cycles} implies that extreme points of  
$\Pi^{Cut}$ have the property given in the lemma.
From Lemma~\ref{l:cycles} we also get:

\begin{corollary} \label{c:bound}
In the case of unit costs there exists a polynomial time algorithm that computes
$x \in \Pi$ such that  the graph $(V,E^{\sf cr})$ of cross edges of positive $x$-value is a forest
and such that and $C_x^{\sf in}=C^{\sf in}$, $C_x^r=C^r$, and $C_x^{\sf cr} \leq \frac{4}{3}C^{\sf cr}$.
\end{corollary}
\begin{proof}
Let $\Pi$ be as in Lemma~\ref{l:cycles} where $a=C^{\sf in}$ and $b=C^r$ and let 
$x$ be an optimal extreme point solution to the LP 
$\min\{\sum_{e \in E}x_e:x \in \Pi\}$.
Let $Q$ be a cycle of cross-edges and $e$ the minimum $x$-value edge in $Q$. 
We update $x$ by adding $x_e$ to each of $x_{e'},x_{e''}$ and setting $x_e=0$.
The increase in the value of $x$ is at most $\frac{1}{3}\sum_{e \in Q}x_e$, and 
it is not hard to verify that $x$ remains a feasible solution.
In this way we can eliminate all cycles, ending with $x \in \Pi$ as required. 
\qed
\end{proof}

\noindent
{\bf Remark.}
Corollary~\ref{c:bound} holds also for arbitrary costs, but in this case the proof
is much more involved. Specifically, we use the following statement, which we do not prove here, 
since it currently has no application: \\
{\em Let $q \geq 3$ and let $c_i,x_i \geq 0$ be reals, $i=0, \ldots,q-1$.
Denote $a_i=c_{i-1}-c_i+c_{i+1}$ where the indices are modulo $k$. Then 
$\sum_{i=0}^{k-1} c_ix_i \geq 3 \cdot \min_{0 \leq i \leq k-1} a_ix_i$.} \\

Let $x$ be as in Corollary~\ref{c:bound} and let $x'$ be an $E^{\sf in }$-up vector of $x$. 
Note that $x' \in \Pi^{Cut}$, since $x \in \Pi^{Cut}$.
We will show how to compute a solution $J$ of size $c(J) \leq x'(E) \leq 2C^{\sf in}+\frac{4}{3}C^{\sf cr}+C^r$. 
While there exists  a pair of edges $e=uv$ and $e'=u'v'$ such that $x'_e,x'_{e'}>0$ 
and $T_{u'v'} \subset T_{uv}$ we do $x'_e \gets x'_e+x'_{e'}$ and $x'_{e'} \gets 0$.
Then $x'$ remains a feasible solution to the {\CLP} without changing the value 
(since we are in the case of unit costs).
Hence we may assume that there is no such pair of edges.
Let $E'$ be the support of $x'$.
If every leaf of $T$ has some cross-edge in $E'$ incident to it, 
then by the assumption above there are no up-edges.
In this case, since $E'$ is a forest, $x_e \geq 1$ for every $e \in E'$ and $E'$ is a solution as required.

Otherwise, there is a leaf $v$ of $T$ such that no cross-edge in $E'$ is incident to $v$.
Then there is a unique up-edge $e$ incident to $v$, and $x'_e \geq 1$.
We take such $e$ into our partial solution, updating $x'$ and $E'$ accordingly.  
Note that some cross-edges may become $r$-edges, but no up-edge can become a cross-edge,
and the set of cross-edges remains a forest.
Applying this as long as such leaf $v$ exists,
we arrive at the previous case, where adding $E'$ to the partial solution
gives a solution as required. This concludes the proof of Lemma~\ref{l:round2}.

%%%%%%%
\subsection{Comparison to the results of Fiorini et al.} \label{ss:FGKS}
%%%%%%%

We need some definitions to compare Corollary~\ref{c:half} to a result of
Fiorini, Gro\ss, K\"{o}nemann, and Sanit\'{a} \cite{FGKS}, that
showed that spider-shaped {\TA} instances can be solved in polynomial time.
Consider the polyhedron $\Pi(b)=\{x:Ax \geq b, x \geq 0\}$ where $A$ is a given integral matrix,
and let $\Pi_I(b)$ be convex hull of the integral points in $\Pi(b)$. 
The {\bf $\{0,\frac{1}{2}\}$-Chv\'{a}tal-Gomory cuts} (see \cite{Gom,Chv,CF}) are inequalities of the form  
$(\la^\T A +\mu^\T) x \geq \lceil \la^ \T b \rceil$, 
for vectors $\la,\mu$ with entries in $\{0,\frac{1}{2}\}$ such that $\la^\T A+\mu^\T$ is an integral vector.

A matrix $A$ is {\bf $2$-regular} if each of its non-singular square submatrices is half-integral.
It is known that $A$ is $2$-regular if and only if the extreme points of $\Pi(b)$ are half-integral for any integral 
vector $b$, and that if $A$ is $2$-regular then 
$P_I(b)$ is described by the $\{0,1/2\}$-Chv\'{a}tal-Gomory cuts \cite{AK04}.
Thus in matrix terms our Corollary~\ref{c:half} implies the following:

\begin{corollary} \label{c:2-regular}
In spider-shaped {\TA} instances, 
the incidence matrix $A$ of the $T$-edges and the paths $\{T_e:e \in E\}$ is $2$-regular.
\end{corollary}

Note that $2$-regularity of $A$ does not imply that the corresponding 
integer program $\min\{c^\T x:x \in \Pi_I(b)\}$ is in P, 
since we have no guarantee that the separation problem for $\{0,1/2\}$-Chv\'{a}tal-Gomory cuts is in P.
% (e.g., as in the {\sc Vertex-Cover} problem).
However, a particular class of $2$-regular matrices has this nice property. 
A matrix $A$ is a {\bf binet matrix} if there exists a square non-singular integer matrix $R$
such that ${\|z\|}_1 \leq 2$ for any column $z$ of $R$ or of $RA$,
where ${\|z\|}_1=\sum_i|z_i|$ is the $L^1$-norm of $z$.
It is known that any binet matrix is $2$-regular, 
but binet matrices have the advantage that 
the separation problem for $\{0,1/2\}$-Chv\'{a}tal-Gomory cuts is in P \cite{AKKP}.
All in all, we have that if $A$ is binet then the integer program
$\min\{c^\T x: x \in \Pi_I(b)\}$ can be solved efficiently,
by a combinatorial algorithm \cite{AKKP}.
The following result, that is stronger than our Corollary~\ref{c:2-regular},
was proved by Fiorini, Gro\ss, K\"{o}nemann \& Sanit\'{a} \cite{FGKS} in parallel to our work;
for completeness of exposition we provide a proof-sketch.

\begin{lemma}[\cite{FGKS}] \label{l:FGKS}
In spider-shaped {\TA} instances, 
the incidence matrix $A$ of the $T$-edges and the paths $\{T_e:e \in E\}$ is~binet.
\end{lemma}
\begin{proof}
For $f \in F$ let $\ch(f)$ denote the set of child $T$-edges of $f$ in $T$.
Define a square matrix $R \in \{-1,0,1\}^{F \times F}$ as follows:
$R_{f,f}=1$,
$R_{f,g}=-1$ if $g \in \ch(f)$, % if $f$ is the parent $T$-edge of $g$, 
and the other entries of $R$ are $0$.
Let $z$ be the column in $R$ of $g \in F$. 
Then $z_g=1$ and if $g$ has a parent $T$-edge $f$ then $z_f=-1$; other entries of $z$ are~$0$.
Thus ${\|z\|}_1 \leq 2$.
We prove by induction on $|F|$ that $R$ is non-singular. 
The case $|F|=1$ is trivial. If $|F| \geq 2$, let $f$ be a leaf $T$-edge. 
The row of $f$ in $R$ has a unique non-zero entry $R_{f,f}=1$. 
Let $T'$ be obtained from $T$ by removing $f$ and the leaf of $f$.
The matrix $R'$ that corresponds to $T'$ is obtained from $R$ by removing the row of $f$
and the column of $f$. By the induction hypothesis,  $\det(R') \neq 0$.
Thus $|\det(R)|=|\det(R')| \neq 0$, implying that $R$ is non-singular.

We now describe the entries of the matrix $RA$. Let $y$ be the row in $R$ of $f \in F$. 
Then $y_f=1$ and $y_g=-1$ for $g \in \ch(f)$; other entries of $y$ are $0$.
Column $e$ in $A$ encodes the path $T_e$, namely, has $1$ for each $T_e$-edge; other entries are $0$.  
Thus 
$$
(RA)_{f,e}=|f \cap T_e|-|\ch(f) \cap T_e| \ .
$$
In particular, if $z$ is the column in $RA$ of $e \in E$ then: 
\begin{itemize}
\item
If $f \in T_e$ then $z_f=1$ if $|\ch(f) \cap T_e|=0$ and $z_f=0$ otherwise.
\item
If $f \notin T_e$ then $z_f=-|\ch(f) \cap T_e|$. % and note that $|\ch(f) \cap T_e| \leq 2$ for
\end{itemize}
Now let $e=uv$ and let $a$ be the least common ancestor of $u,v$.
Consider two cases, in which we indicate only non-zero entries of $z$.
If $a \in \{u,v\}$ ($e$ is an up-edge), say $a=v$, then
$z_f=1$ if $f$ is the parent $T$-edge of $u$ and $z_f=-1$ if $a \neq r$ and $f$ is the parent $T$-edge of $v$.
% and $z_f=0$ otherwise.
If $a \notin \{u,v\}$ then
$z_f=1$ if $f$ is the parent $T$-edge of $u$ or of $v$,
and $z_f=-2$ if $f$ is the parent $T$-edge of $a$;
% and $z_f=0$ otherwise.
however, in a spider-shaped {\TA} instance we cannot have $z_f=-2$, 
since if $e$ is a cross edge then $a=r$ and thus $a$ has no parent $T$-edge.
Consequently, in both cases ${\|z\|}_1 \leq 2$.
\qed
\end{proof}

By a result of \cite{AKKP} (an integer program $\min\{c^\T x: x \in \Pi_I(b)\}$ is in P if $A$ is binet), 
Lemma~\ref{l:FGKS} immediately implies:

\begin{corollary} [\cite{FGKS}] \label{c:FGKS}
Spider-shaped {\TA} instances admit a polynomial  time algorithm.
\end{corollary}

In \cite{FGKS} it is also provided a direct simple proof that the problem of 
separating the $\{0,1/2\}$-Chv\'{a}tal-Gomory cuts of the {\CLP} is in P.
Combining this with our Corollary~\ref{c:2-regular} and a result of \cite{AK04} 
($P_I(b)$ is described by the $\{0,1/2\}$-Chv\'{a}tal-Gomory cuts if $A$ is $2$-regular),
also enables to deduce Corollary~\ref{c:FGKS}.

Fiorini et. al. \cite{FGKS} considered the {\sc Odd-Cut $k$-Bundle LP} obtained by adding
to the $k$-{\sc Bundle LP} of \cite{A} the ``odd-cuts'' (the $\{0,1/2\}$-Chv\'{a}tal-Gomory cuts).
They showed that this LP is compatible with the decomposition of \cite{A},
namely, that if $x$ a feasible solution to this LP, then for any $k$-bundle $B$ 
the restriction of $x$ to $\psi(B)$ is a feasible solution to this LP on $B$.
Since each $k$-branch is a $k$-bundle, 
the more compact {\sc Odd-Cut $k$-Branch LP} is also compatible with \cite{A} decomposition.
As was mentioned in the Introduction, combining the \cite{FGKS} and our paper techniques gives: \\
{\em For any $1 \leq \la \leq k-1$, {\TA} admits a $4^k \cdot poly(n)$ time algorithm that 
computes a solution of cost at most 
$\frac{3}{2}+\frac{2\la M}{k-\la M}+\frac{2}{\la}$ times the optimal value of the 
{\sc Odd-Cut $k$-Branch LP}}. \\
Let us briefly describe the modifications needed for this combined result.
\begin{itemize}
\item
Lemma~\ref{l:thick} is used in the same way as before, namely, just to cover 
by cost $\frac{2}{\la}c^\T x$ the $\la$-thick edges uncovered by the main algorithm.
\item
Recall that \cite{FGKS} showed that separating the odd-cuts 
% (the $\{0,1/2\}$-Chv\'{a}tal-Gomory cuts of the {\CLP}) 
is in P.
The new Lemma~\ref{l:round1} would state that given $x^* \in \mathbb{R}^E$,
there exists a $4^k \cdot poly(n)$ time algorithm that 
either finds a $k$-branch inequality or an odd-cut inequality violated by $x^*$, 
or returns an integral solution of cost at most $C^{\sf in}+2C^{\sf cr}+C^r$.
\item
Lemma~\ref{l:round2} will be replaced by a result of \cite{FGKS} that a solution 
of cost $2C^{\sf in}+C^{\sf cr}+C^r$ can be computed in polynomial time.
\item
In an improved version of Corollary~\ref{c:main} one gets that if no violated inequality 
is found then $c(J_S) \leq \sum_{e \in \ga(S)}c_ex_e+\sum_{e \in \psi(f)}c_ex_e$.
And then, the same calculations as after Algorithm~\ref{alg:rounding} give
$\frac{c(J_S)}{\Delta(c^\T x)} \leq \frac{3}{2}+\frac{2\la M}{k-\la M}$. 
\end{itemize}

Let us now illustrate another application of Corollary~\ref{c:FGKS}.

\begin{lemma} \label{l:diam}
{\TA} admits ratio $3/2$ for trees of diameter $\leq 5$.
\end{lemma}
\begin{proof}
The case $\diam(T)=5$ is reduced to the case $\diam(T) \leq 4$ by ``guessing'' some optimal solution edge that covers the central $T$-edge. 
So assume that $\diam(T) \leq 4$. Let $r$ be a center of $T$. Let $r$ be a center of $T$. 
Fix some optimal solution and let $C^{\sf in}$ and $C^{\sf cr}$ denote 
the fractional cost of in-edges and cross-edges in this solution.
As before, apply the following two procedures.
\begin{enumerate}
\item
Each branch $B$ hanging on $r$ is a tree of diameter $\leq 3$, 
hence an optimal cover $J_B$ of $B$ can be computed in polynomial time.
The union of the edge sets $J_B$ gives a solution of cost at most $C^{\sf in}+2C^{\sf cr}$. 
\item
Compute an optimal solution of the spider-shaped instance obtained 
by removing all non-up in-edges using Lemma~\ref{l:FGKS}; the cost of this solution is $2C^{\sf in}+C^{\sf cr}$.
\end{enumerate}
Choosing the better among the two computed solutions gives a solution of cost at most
$\min\{C^{\sf in}+2C^{\sf cr},2C^{\sf in}+C^{\sf cr}\}$, while the optimal solution cost is $C^{\sf in}+C^{\sf cr}$.
It is easy to see that the approximation ratio is bounded by $3/2$; 
if $C^{\sf in} \leq C^{\sf cr}$ then $C^{\sf in}+2C^{\sf cr} \leq \frac{3}{2}(C^{\sf in}+C^{\sf cr})$,
while if $C^{\sf in} > C^{\sf cr}$ then $2C^{\sf in}+C^{\sf cr} < \frac{3}{2}(C^{\sf in}+C^{\sf cr})$.
\qed
\end{proof}

Lemma~\ref{l:diam} can be used further to obtain ratio $9/5$ for trees of diameter $\leq 7$.
As before, we can reduce the case $\diam(T)=7$ to the case $\diam(T) \leq 6$ by guessing some optimal solution 
edge that covers the central $T$-edge. We compute a $3/2$-approximate cover of each branch, which gives a solution 
of cost at most $\frac{3}{2}(C^{\sf in}+2C^{\sf cr})$. We also compute a solution of cost at most $2C^{\sf in}+C^{\sf cr}$
as before, using Corollary~\ref{c:FGKS}. The worse case is when these two bounds are equal, namely, when 
$C^{\sf in}=4C^{\sf cr}$. In this case we get that 
$\frac{2C^{\sf in}+C^{\sf cr}}{C^{\sf in}+C^{\sf cr}} = 1+ \frac{C^{\sf in}}{C^{\sf in}+C^{\sf cr}}=1+\frac{4}{4+1}=\frac{9}{5}$.
In a similar way, one can further obtain ratio better than $2$ when $\diam(T) \leq 9$, and so on,
but the ratio approaches $2$ when the diameter becomes higher. \\

We note that the effort in proving Lemma~\ref{l:FGKS} of \cite{FGKS} and our 
Corollary~\ref{c:2-regular} is roughly the same. 
However, the result in Lemma~\ref{l:FGKS} of \cite{FGKS} is more general
and thus enables to obtain easily results for related problems, as we illustrate below.
Note however that our result not only substantially simplifies and reduces the time complexity of algorithms based on
the approach of Adjiashvili \cite{A}, but also qualitatively extends the range of costs for which a ratio better than $2$ can be achieved. 
Moreover, the proof idea of Corollary~\ref{c:half} 
might be useful for half-integral network design problems for which the corresponding matrix is not binet.

It is known that if $A$ is binet then also the problem
of minimizing $c^\T x$ over $\{x \in \Pi_I(b):p \leq x \leq q\}$ can be solved in
polynomial time for any integer vectors $p$ and $q$ \cite{AK04,AKKP}.
Now consider the following generalization
of {\TA}, which we call the {\sc Generalized Tree Augmentation} problem. 
Here we are also given demands $\{b_f:f \in F\}$ on the $T$-edges,
and require that at least $b_f$ edges will cover every $T$-edge $f \in F$;
we also require that for every edge $e \in E$ at most $q$ copies of $e$ are selected. 
Then from Lemma~\ref{l:FGKS} of \cite{FGKS} and \cite{AK04,AKKP} one can deduce the following (the proof is omitted):

\begin{corollary} \label{c:FGKS2}
Spider-shaped {\sc Generalized Tree Augmentation} instan\-ces admit a polynomial  time algorithm.
\end{corollary}

%%%%%%%%%%%%%%%%%%%%%%
\section{Bound on the integrality gap of the {\CLP} (Theorem~\ref{t:gap})} \label{s:gap}
%%%%%%%%%%%%%%%%%%%%%%

Let us write the (unit costs) {\CLP} as well as its dual LP explicitly:
\[
\begin{array}{lllllllll} \ \ \ \ \ \ \ 
& \min             & \di \   \sum_{e \in E} x_e                      &  \ \ \ \ \ \ \ \ \ \ \ \ \ \ \ \ \ \ \ \                   
& \max            & \di \   \sum_{f \in F} y_f                        &                         \\
& \mbox{s.t.}  & \di     \sum_{e \in \psi(f)} x_e  \geq 1  & \forall f \in F 
& \mbox{s.t.}  & \di     \sum_{\psi(f) \ni e}y_f \leq 1      & \forall e \in E   \\
&                      & \di \ \ x_e \geq 0                                   & \forall e \in E 
&                      & \di \ \ y_f \geq 0                                    & \forall f \in F  
\end{array}
\]

To prove that the integrality gap of the {\CLP}  is at most $28/15$
we will show that a simplified version from \cite{KN16} of 
the algorithm of \cite{EFKN-TALG} has the desired performance.
For the analysis, we will use the dual fitting method. 
We will show how to construct a (possibly infeasible) dual solution $y \in \mathbb{R}^F_+$, 
that has the following two properties: \\
\noindent
{\bf Property~1.} \ 
$y$ fully pays for the constructed solution $J$, namely, $|J| \leq \sum_{f \in F} y_f$. \\
{\bf Property~2.} \ 
$y$ may violate the dual constraints by a factor of at most  $\rho=28/15$.  \\
% namely, $\sum_{\psi(f) \ni e}y_f \leq \rho$ for all $e \in E$. \\
\noindent
From the second property we get that $y/\rho$ is a feasible dual solution,
hence by weak duality the value of $y$ is at most $\rho$ times the optimal value of the {\CLP}.
Combining with the first property we get that $|J|$ is at most $\rho$ times the optimal value of the {\CLP}.

% In this section, for a set $J$ of edges let $T/J$ denote the tree (or the {\TA} instance) 
% obtained by contracting every $2$-edge-connected component of $T \cup J$ into a single node; 
% we also refer to this contraction as the contraction of the edges in $J$.
% Note that $J$ covers $T$ if and only if $T/J$ is a single node.
% We say that $J'$ is an {\bf exact cover} of a subtree $T'$ of $T/J$
% if the set of edges of $T/J$ covered by $J'$ equals the set of edges of $T'$.

The algorithm % of \cite{EFKN-TALG,KN16} 
iteratively finds a pair $T',J'$ where 
$T'$ is a subtree of the current tree and $J'$ covers $T'$, contracts $T'$, and adds $J'$ to $J$.
We refer to nodes created by contractions as {\bf compound nodes} and denote by $C$ 
the set of non-leaf compound nodes of the current tree. 
Non-compound nodes are referred to as {\bf original nodes}. 
For technical reasons, the root $r$ is considered as a compound node. % hence initially $C=\{r\}$.
Whenever $T'$ contains the root of $T$, 
the new compound node becomes the root of the new tree.

To identify a pair $T',J'$ as above, the algorithm maintains a matching $M$ on the original leaves.
We denote by $U$ the leaves of the current tree unmatched by $M$. 
A subtree $T'$ of $T$ is {\bf $M$-compatible} if for any $bb' \in M$
either both $b,b'$ belong to $T'$ or none of $b, b'$ belongs to $T'$;
in this case we will also say that a contraction of $T'$ is $M$-compatible.
Assuming all compound nodes were created by $M$-compatible contractions, 
then the following type of contractions is also $M$-compatible.

\begin{definition} [greedy contraction]
Adding to the partial solution $J$ an edge $e$ with both endnodes in $U$ and contracting $T_e$ 
is called a {\bf greedy contraction}.
\end{definition}

Given a complete rooted $M$-compatible subtree $T'$ of $T$ we use the notation:
\begin{itemize}
\item
$M' = M(T')$ is the set of edges in $M$ with both endnodes in $T'$.
\item
$U'= U(T')$ is the set of unmatched leaves of $T'$.
\item
$C'= C(T')$ is the set of non-leaf compound nodes of $T'$.
\end{itemize}

\begin{definition} [semi-closed tree]
Let $T'$ be a complete rooted subtree of $T$.
For a subset $A$ of nodes of $T'$ we say that $T'$ of is {\bf $A$-closed} if 
there is no edge from $A$ to a node outside $T'$, and $T'$ is {\bf $A$-open} otherwise.
Given a matching $M$ on the leaves of $T$, we say that $T'$ is {\bf semi-closed} if it is $M$-compatible and $U'$-closed. 
% $T'$ is {\bf minimally semi-closed} if $T'$ is semi-closed but any proper subtree of $T'$ is not semi-closed.
\end{definition}

The following definition characterizes semi-closed subtrees that we want to avoid.
We will say that $T'$ with $3$ leaves is of type~(i) if it has two nodes with exactly two children each 
(see the node $w$ and its parent in Fig.~\ref{f:def}(i)) 
and $T'$ is of type~(ii) otherwise (see Fig.~\ref{f:def}(ii)).

\begin{figure} \centering
\includegraphics{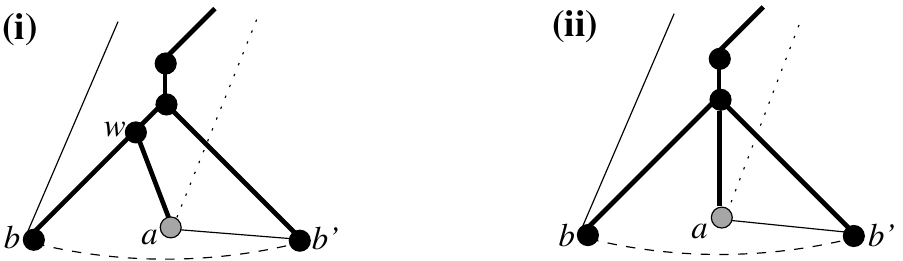}
\caption{Dangerous trees. Here and in subsequent figures $T$-edges are shown by bold lines, 
edges in $M$ by dashed lines, some other existing edges by thin solid lines, 
and edges that cannot exist by dotted lines. 
Nodes that must be original are shown by 
black circles, while nodes that may be compound nodes are shown by gray circles.
Some of the edges may be paths, possibly of length~$0$.
A dangerous tree of type~(i) has two nodes with exactly $2$ children each, and 
contracting the path between these two nodes results in a dangerous tree of type~(ii).}
\label{f:def}
\end{figure}

\begin{definition}[dangerous semi-closed tree]
A semi-closed subtree $T'$ of $T$ is {\bf dangerous} if it is as in Fig.~\ref{f:def}. 
Namely, $|M'|=1$, $|U'|=1$, $|C'|=0$, and if $a$ is the leaf of $T'$ unmatched by $M$ then:
$T'$ is $a$-closed and there exists an ordering $b,b'$ of the matched leaves of $T'$
such that $ab' \in E$, the contraction of $ab'$ does not create
a new leaf, and $T'$ is $b$-open. 
\end{definition}

\begin{definition}[twin-edge, stem] 
Let $L$ denote the set of leaves of $T$.
An edge on $L$ is a {\bf twin-edge} if its contraction results in a new leaf.
The least common ancestor of the endnodes of a twin-edge is a {\bf stem}.
\end{definition} 

In \cite{EFKN-TALG} the following is proved: 

\begin{lemma}[\cite{EFKN-TALG}] \label{l:B}
Suppose that $M$ has no twin-edges and that the current tree $T$ was obtained from the initial tree by 
sequentially applying a greedy contraction or a semi-closed tree contraction, and that $T$ has no greedy contraction. 
Then there exists a polynomial time algorithm that finds a non-dangerous semi-closed subtree $T'$ of $T$ 
and a cover $J'$ of $T'$ of size $|J'|=|M'|+|U'|$. 
\end{lemma}

Let $L(M)$ denote the set of leaves matched by $M$. The algorithm is as follows:

\medskip

\begin{algorithm}[H]
\caption{{\sc Iterative-Contraction}$(T=(V,F),E)$} \label{alg:plain}
{\bf initialize:}                    $M \gets$ maximal matching on $L$ among non twin-edges \newline
\hphantom{\bf Initialize:} $J \gets$   maximal matching on $L \sem L(M)$                        \\ 
contract every link in $J$                                                                                                                              \\            
\While{\em $T$ has at least $2$ nodes} 
{
exhaust greedy contractions                                                                                                                         \\
if $T$ has at least $2$ nodes then for $T',J'$ as in Lemma~\ref{l:B} do:                                                    \newline 
$J \gets J \cup J'$, $T \gets T/T'$     
}
\Return{$J$}
\end{algorithm}

\medskip 

We now describe how to construct $y$ satisfying Properties 1 and 2 as above.
For simplicity of exposition let us use the notation $y_v$ % and $y_{T'}$ 
to denote the dual variable of the parent $T$-edge of $v$. % and of $T'$, respectively.
With this notation, Algorithm~\ref{alg:duals} incorporates into Algorithm~\ref{alg:plain} 
the steps of the construction of the dual (possibly infeasible) solution $y$.

\medskip 

\begin{algorithm}[H]
\caption{{\sc Dual-Construction}$(T=(V,F),E)$} \label{alg:duals}
{\bf initialize:}                    $M \gets$ maximal matching on $L$ among non twin-edges           \newline
\hphantom{\bf Initialize:} $J \gets$ maximal matching on $L \sem L(M)$ (see Fig.~\ref{f:init})  \newline 
\hspace*{2.5cm} $\bullet$ $y_v \gets 1$ \ \ \ \ \ \ \ \ \ if $v \in L \sem (M \cup J)$                                           \newline
\hspace*{2.5cm} $\bullet$ $y_v \gets 4/5$ \ \ \ \  \ \    if $v \in L(M)$                                                                \newline 
\hspace*{2.5cm} $\bullet$ $y_v \gets 14/15$ \ \ \         if $v \in L(J)$                                                                  \newline
\hspace*{2.5cm} $\bullet$ $y_v \gets 2/15$ \hspace{0.4cm} if $v$ is a stem of an edge in $J$                         \\
contract every link in $J$                                                                                                                                         \\  
\While{\em $T$ has at least $2$ nodes} {exhaust greedy contractions                                                               \\
if $T$ has at least $2$ nodes then for $T',J'$ as in {\bf Lemma~\ref{l:B}} do:                                                       \newline 
$J \gets J \cup J'$, $T \gets T/T'$                                                                                                                              \newline
\hspace*{0.35cm} {\bf Case 1:} $|C'|=0$ and either: $|M'|=0$ or $|M'|=1,|U'| \geq 2$                                    \newline
\hspace*{1.8cm} $\bullet$ update $y$ as shown in {\bf Fig.~\ref{f:easy}}                                                             \newline
% \hspace*{1.8cm} $\bullet$ $y_{T'} \gets 2/5$                                                                                                          \newline
% \hspace*{1.8cm} $\bullet$ $y_v \gets y_v+2/5$ if $v \in U'$ and $y_v \gets y_v-2/5$ if $v \in L' \sem U'$        \newline
\hspace*{0.35cm} {\bf Case 2:} $|C'|=0$ and $|M'|=|U'|=1$                                                                                \newline 
\hspace*{1.8cm} $\bullet$ update $y$ as shown in {\bf Fig.~\ref{f:hard}}                                                      
}
\Return{$J$}
\end{algorithm}

\medskip 

\begin{figure} \centering
\includegraphics{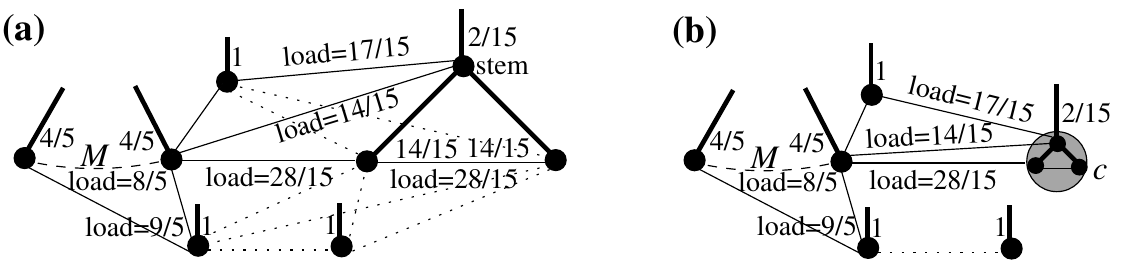}
\caption{ $T$-edges are shown by bold lines, 
edges in $M$ by dashed lines, some other existing edges by thin solid lines, 
and edges that cannot exist by dotted lines. 
(a) Initial duals at step~1 of Algorithm~\ref{alg:duals} and the initial loads. 
Here there is one stem and $|M|=1$.
(b) After contracting the twin-edge at step~2, the new compound node $c$ has credit $1$.}
\label{f:init}
\end{figure}

\begin{figure} \centering
\includegraphics{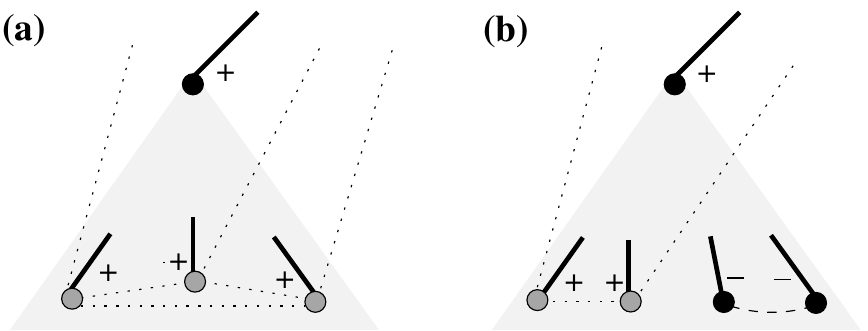}
\caption{Duals updates in Case~1 of Algorithm~\ref{alg:duals}. 
(a) $|M'|=0$; here ``$+$'' means increasing the dual variable by $1/2$.
(b) $|M'|=1$, $|U'| \geq 2$; here ``$+$'' means increasing the dual variable by $2/5$ and ``$-$'' means 
decreasing the dual variable by $2/5$.}
\label{f:easy}
\end{figure}

\begin{figure} 
\centering
\includegraphics{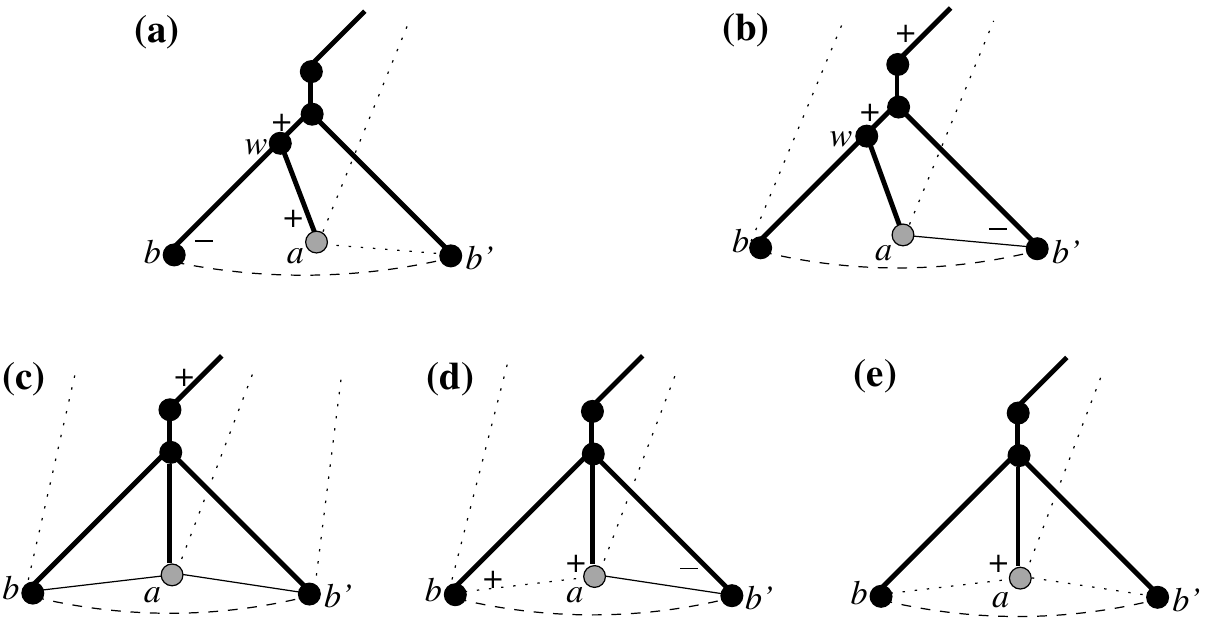}
\caption{Non-dangerous trees with $|M'|=|U'|=1$ and duals updates in Case~2 of Algorithm~\ref{alg:duals}. 
Here ``$+$'' means increasing the dual variable by $2/5$ and ``$-$'' means 
decreasing the dual variable by $2/5$. 
% Note that this includes all non-dangerous trees with $|M'|=|U'|=1$. 
All trees are $a$-closed. The trees in (a,b) are non-dangerous trees of type~(i),
and the trees in (c,d,e) are non-dangerous trees of type~(ii).
In (a) the edge $ab'$ is missing and in (b) $ab'$ is present and $T'$ is $b$-closed. 
In (c) both edges $ab$ and $ab'$ are present, hence to be non-dangerous 
the tree must be both $b'$-closed and $b$-closed. 
In (d) $ab'$ is present hence the tree must be $b$-closed; 
the case when $ab$ present and the tree is $b'$-closed is identical. 
In (e) both $ab$ and $ab'$ are missing.}
\label{f:hard}
\end{figure}

We now define certain quantities that will help us to prove that 
at the end of the algorithm $|J| \leq \sum_{f \in F}y_f$ and that $y$ 
violates the dual constraints by a factor of at most $28/15$.

\begin{definition}[load of an edge]
Given $y \in \mathbb{R}^F_+$ and an edge $e \in E$, the {\bf load $\si(e)$ of $e$} is 
the sum of the dual variables in the constraint of $e$ in the dual LP, namely
$\si(e)=\sum_{\psi(f) \ni e} y_f$.
\end{definition}

\begin{definition}[credit of a node]
Consider a constructed dual solution $y$ and a node $c$ of $T$ during the algorithm,
where $c$ is obtained by contracting the (possibly trivial) subtree $S$ of $T$. 
The {\bf credit $\pi(c)$} is defined as follows.
Let $\pi'(c)$ be the sum of the dual variables $y$ of the edges of $S$ and the parent edge of $c$
minus the number of edges used by the algorithm to contract $S$ into $c$.
Then $\pi(c)=\pi'(c)+1$ if $r \in S$ and $\pi(c)=\pi'(c)$ otherwise. 
\end{definition}

Our goal is to prove that at the end of the algorithm $\si(e) \leq 28/15$ for all $e \in E$,
and that the unique node of $T$ has credit at least $1$.
% We will use the notation $e \sim uv$ meaning that $e$ connects 
% the two nodes $u$ and $v$ of the current tree $T$,
% while $e=uv$ means that $u$ and $v$ are both original nodes.
% (if $e \sim uv$ and $e=uv$ then $u,v$ are both original nodes). 
For an edge $e$ that connects nodes $u,v$ of the current tree $T$
the {\bf level} $\ell(e)$ of $e$ (w.r.t. the current tree $T$) 
is the number of compound nodes and original leaves (of the current tree $T$) in $\{u,v\}$. 
Clearly, $\ell(e) \in \{0,1,2\}$ and note that if both endnodes of $e$ 
lie in the same compound node then $e$ is a loop and $\ell(e)=2$.

\begin{lemma} \label{l:levels}
At the end of step~2 of Algorithm~\ref{alg:duals}, 
and then at the end of every iteration in the ``while'' loop, the following holds.
\begin{itemize}
\item[{\em (i)}]
$\pi(c) \geq 1$ if $c$ is an unmatched leaf or a compound node of $T$.
\item[{\em (ii)}]
For any edge $e$:
\begin{itemize}
\item[$\bullet$]  
$\si(e) \leq 28/15$ \ if $\ell(e)=2$.                                                                       
\item[$\bullet$] 
$\si(e) \leq 16/15$ \ if $\ell(e)=1$.                                              
\item[$\bullet$] 
$\si(e)=0$   \hspace*{0.7cm} if $\ell(e)=0$. % otherwise 
\end{itemize}
 \end{itemize}
\end{lemma}
\begin{proof}
It is easy to see that the statement holds at the end of step~2,~see~Fig.~\ref{f:init}.
We will prove by induction that the statement continues to hold after each contraction step of the while-loop. 
Let us consider such contraction step that resulted in a new compound node $c$ and denote by 
$\si', \ell', \pi'$ the new values of $\si, \ell, \pi$ after the contraction. 
By the induction hypothesis 
$\si,\ell,\pi$ satisfy properties (i) and (ii) above, and we prove that 
$\si',\ell',\pi'$ satisfy (i) and (ii) as well.

For (i) it is sufficient to prove that $\pi'(c) \geq 1$, as $\pi'=\pi$ for other nodes.
Consider a greedy contraction with an edge $e$ connecting two unmatched leaves $u$ and $v$.
By the induction hypothesis, $\pi(u),\pi(v) \geq 1$.
Thus $\pi'(c) \geq \pi(u)+\pi(v)-1 \geq 1$.
Now suppose that a semi-closed tree $T'$ was contracted into $c$.
Let $\Delta(y)$ denote the increase in the value of $y$ during the contraction step and note that 
$$
\pi'(c) \geq \left(\pi(C')+\frac{8}{5}|M'|+|U'|\right)-(|M'|+|U'|) +\Delta(y) \geq |C'|+\frac{3}{5}|M'| +\Delta(y) \ .
$$
If $|C'| \geq 1$ or $|M'| \geq 2$ then $\pi'(c) \geq |C'|+\frac{3}{5}|M'| \geq 1$. 
If $|M'|=0$ (Fig.~\ref{f:easy}(a)) then $\pi'(c) \geq \Delta(y) \geq \frac{1}{2}(|U'|+1) \geq 1$.
If $|M'|=1$ then $\Delta(y) = \frac{2}{5}$, since in all cases in Figures \ref{f:easy}(b) and \ref{f:hard},
the number of ``$+$'' signs is larger by one than the number of ``$-$'' signs;
thus $\pi'(c) \geq \frac{3}{5}|M'| +\frac{2}{5} \geq 1$.
In all cases $\pi'(c) \geq 1$, as required.

We now show that property (ii) holds. 
Note that if $\si'(e)=\si(e)$ then (ii) continues to hold for $e$, since contractions can only increase the edge level
and since the bounds in (ii) are increasing with the level. 
Thus we only need to consider the cases when we change the dual variables, namely, when
a semi-closed tree $T'$ was contracted into $c$; these are the cases given in Figures \ref{f:easy} and \ref{f:hard}. 

It is sufficient to consider edges with at least one endnode in $T'$, as $\si'=\si$ and $\ell'=\ell$ holds for other edges.
Let $e$ be an edge that has an endnode in $T'$.
Let $q(e)$ denote the number of ``$+$'' signs minus the number of ``$-$'' signs 
in Figures \ref{f:easy} and \ref{f:hard} along the path $T_e$;
we have $\si'(e)-\si(e)=\frac{1}{2}q(e)$ in Fig.~\ref{f:easy}(a) and $\si'(e)-\si(e)=\frac{2}{5}q(e)$ in all other cases.
One can verify that $q(e) \leq 0$ if $e$ connects a leaf of $T'$ to another leaf of $T'$ or to a node outside $T'$.
Thus it remains to consider the case when $e$ is incident to a non-leaf node of $T'$.
Then $\ell'(e)>\ell(e)$, since $|C'|=0$. 
One can verify that $q(e) \leq 1$, except one case -- $q(e)=2$ if in Fig.~\ref{f:hard}(a) 
$e$ connects the leaf $a$ to a node $v$ in $T'$ that is an ancestor of $w$;
this tight case is the one that determined our initial assignment of dual variables.
In all cases we have $\si'(e)-\si(e) \leq \frac{4}{5}$, which equals 
the minimum difference $\frac{28-16}{15}$ in the bounds in (ii) due to an increase of an edge level.
This concludes the proof of (ii) and of the lemma.
\qed
\end{proof}

%%%%%%%%%%%%%%%%%%%%%%%
\section{Integrality gap of the 3-{\BLP}} \label{s:gap'}
%%%%%%%%%%%%%%%%%%%%%%%

The following simple LP-relaxation was suggested by the author several years before \cite{A} and \cite{FGKS}.
Let us call an odd size set $B$ of edges of $T$ a {\bf bunch} if no $3$ edges in $B$ lie
on the same path in $T$. 
Let $\BB$ denote the set of bunches in $T$. 
For every $B \in \BB$ at least $w_B:=(|B|+1)/2$ edges are needed to cover $B$.
The corresponding {\sc Bunch-LP} and its dual LP are:
\[
\begin{array}{lllllllll} \ \ \ \ \ \ \ 
& \min             & \di \   \sum_{e \in E} x_e                      &  \ \ \ \ \ \ \ \ \ \ \ \ \ \ \ \ \ \ \ \                   
& \max            & \di \   \sum_{B \in \BB} w_B y_B                                &                         \\
& \mbox{s.t.}  & \di     \sum_{e \in \psi(B)} x_e  \geq w_B  & \forall B \in \BB 
& \mbox{s.t.}  & \di     \sum_{\psi(B) \ni e}y_B \leq 1      & \forall e \in E   \\
&                      & \di \ \ x_e \geq 0                                   & \forall e \in E 
&                      & \di \ \ y_B \geq 0                                    & \forall B \in \BB  
\end{array}
\]

A {\bf $k$-bunch} is a bunch of size $k$. 
Let $k$-{\BLP} be the restriction of the {\BLP} to bunches of size $\leq k$.
Note that $1$-{\BLP} is just the {\CLP}, and that Theorem~\ref{t:main} says that the 
integrality gap of the $1$-{\BLP} is at most $28/15$. 
We can easily prove a better bound for the $3$-{\BLP}. 
 
\begin{theorem} \label{t:gap'}
For unit costs, the integrality gap of the $3$-{\BLP} is at most $7/4$.
\end{theorem}
\begin{proof}
We use the same algorithm as before, but define the dual variables differently.
% A stem of an edge in $J$ will be called a {\bf $J$-stem}.
In the initialization step we set (see Fig.~\ref{f:init'}):
\begin{itemize}
\item[$\bullet$] 
$y_v \gets 1$ \ \ \ \ \ if $v \in L \sem L(M \cup J)$
\item[$\bullet$]  
$y_v \gets 3/4$ \ \ if $v \in L(M)$
\item[$\bullet$]  
$y_v \gets 1/2$ \ \ if $v \in L(J)$                                                                      
\item[$\bullet$]  $y_B \gets 1/2$  \ if $B$ is the $3$-bunch of % the $T$-edges incident to 
a stem of an edge in $J$ 
\end{itemize}
In the updates of the dual variables in Figures \ref{f:easy} and \ref{f:hard},
``$+$'' and ``$-$'' means increasing and decreasing the dual variable by $1/2$, respectively, with one exception:
in Fig.~\ref{f:hard}(a) the updates are $y_b \gets y_b-1/2$ and
$y_B \gets 1/2$, where $B$ is the $3$-bunch formed by the parent $T$-edges of $a, b,w$.
Similarly to Lemma~\ref{l:levels} we prove that after step~2 the following holds:
\begin{itemize}
\item[(i)]
$\pi(c) \geq 1$ if $c$ is an unmatched leaf or a compound node of $T$.
\item[(ii)]
For any edge $e$:
$\si(e) \leq 7/4$ if $\ell(e)=2$,
$\si(e) \leq 1$ if $\ell(e)=1$, and 
$\si(e)=0$           if $\ell(e)=0$.
\end{itemize}

\begin{figure} \centering
\includegraphics{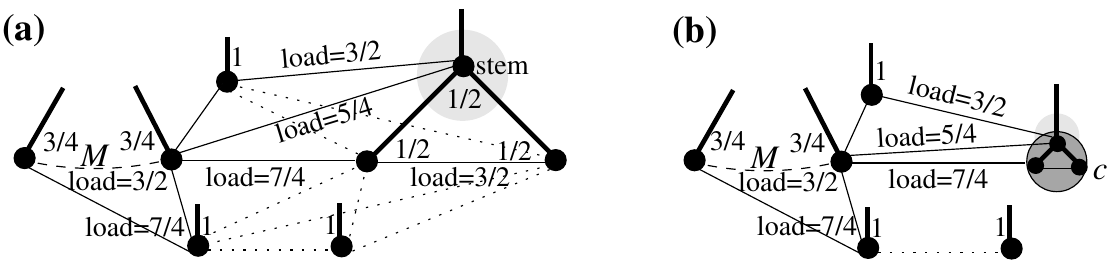}
\caption{Initial duals of the $3$-{\BLP} and the initial loads. 
The $3$-bunch of the edges incident to the stem is shown by a light gray circle.}
\label{f:init'}
\end{figure}

It is easy to see that the statement holds at the end of step~2,~see~Fig.~\ref{f:init'};
note that after step~2 the edge with load $5/4$ has level $2$.
As in Lemma~\ref{l:levels} we continue by induction while using the same notation, 
but focus only on the arguments that are different from the ones in Lemma~\ref{l:levels}.

Suppose that a semi-closed tree $T'$ was contracted into a compound node $c$. Then
$$
\pi'(c) \geq \left(\pi(C')+\frac{3}{2}|M'|+|U'|\right)-(|M'|+|U'|) +\Delta(y) \geq |C'|+\frac{1}{2}|M'| +\Delta(y) \ .
$$
If $|C'| \geq 1$ or $|M'| \geq 2$ then $\pi'(c) \geq |C'|+\frac{1}{2}|M'| \geq 1$. 
If $|M'|=0$ (Fig.~\ref{f:easy}(a)) then $\pi'(c) \geq \Delta(y) \geq \frac{1}{2}(|U'|+1) \geq 1$.
If $|M'|=1$ then $\Delta(y) = \frac{1}{2}$ and thus $\pi'(c) \geq \frac{1}{2}|M'| +\Delta(y) \geq 1$;
this is since in each one of the cases in Figures \ref{f:easy}(b) and \ref{f:hard}(b,c,d,e)
the number of ``$+$'' signs is larger by one than the number of ``$-$'' signs,
while in the case in Fig.~\ref{f:hard}(a) we gain $2 \cdot \frac{1}{2}=1$ when increasing by $\frac{1}{2}$ 
the dual variable of a $3$-bunch, and loose just $\frac{1}{2}$ by decreasing $y_b$ by $\frac{1}{2}$.
In all cases we have $\pi'(c) \geq 1$, as required.

We now show that property (ii) holds. 
Consider a semi-closed tree $T'$ was contracted into $c$ and an edge $e$ with at least one endnode in $T'$.
Note that now $\frac{3}{4}$ is the minimum difference  in the bounds in (ii) due to an increase of an edge level.

Let us consider the case in Fig.~\ref{f:hard}(a).
If $e$ is incident to $b$ or if $e=vb'$ for some $v \in T'$ then $\si'(e) \leq \si(e)$.
In all the other cases we have $\ell'(e) > \ell(e)$ and $\si'(e) - \si(e) \leq \frac{1}{2}<\frac{3}{4}$.
Hence the induction step holds in this case.

For the other cases, as before, let $q(e)$ denote the number of ``$+$'' signs minus the number of ``$-$'' signs 
in Figures \ref{f:easy} and \ref{f:hard}(b,c,d,e) along the path $T_e$;
we have $\si'(e)-\si(e)=\frac{1}{2}q(e)$ in all cases.
One can verify that $q(e) \leq 0$ if $e$ connects a leaf of $T'$ to another leaf of $T'$ or to a node outside $T'$.
If $e$ is incident to a non-leaf node of $T'$ then $\ell'(e)>\ell(e)$ and $q(e) \leq 1$,
which implies $\si'(e)-\si(e) \leq \frac{1}{2}$.
This concludes the proof of (ii) and of the lemma.
\qed
\end{proof}

\section{Conclusions}
In this paper we presented an improved algorithm for {\TA}, based on the idea of Adjiashvili \cite{A}.
A minor improvement is that the algorithm is simpler, as it avoids 
a technical discussion on so called ``early compound nodes'', see \cite{A} and \cite{FGKS}.
A more important improvement is in the running time -- $4^k poly(n)$ instead of $n^{k^{O(1)}}$, 
where $k=\Theta(M/\eps^2)$.
This allows ratio better than $2$ also for logarithmic costs, and not only costs bounded by a constant.
These two improvements are based, among others, on a more compact and simpler LP for the problem. 
Another important improvement is in the ratio -- $\frac{12}{7}+\eps$ instead of $1.96418+\eps$ in \cite{A}.
This algorithm is based on a combinatorial result for spider-shaped {\TA} instances.
% when all in-edges are up-edges. 
We showed that for spider-shaped instances, the extreme points of the {\sc Cut-Polyhedron} are half-integral, 
and thus {\TA} on such instances can be approximated within $4/3$.
As was mentioned, a related recent result of \cite{FGKS} shows that for spider-shaped instances,
augmenting the {\CLP} by $\{0,\frac{1}{2}\}$-Chv\'{a}tal-Gomory Cuts gives 
an integral polyhedron and that such instances can be solved optimally in polynomial time.
Overall we get that spider-shaped instances behave as ``star-instances'' -- when $T$ is a star 
(this is essentially the {\sc Edge-Cover} problem): the extreme points of the {\CLP} are half-integral, 
while augmenting it by $\{0,\frac{1}{2}\}$-Chv\'{a}tal-Gomory Cuts gives an integral polyhedron.
The description of the $\{0,\frac{1}{2}\}$-Chv\'{a}tal-Gomory Cuts in \cite{FGKS} is somewhat complicated,
and a natural question is whether using the simpler {\sc Bunch-LP} gives the same result.
This is so when $T$ is a star, c.f. \cite{Sch} where an equivalent {\sc Edge-Cover} problem is considered.

Our second main result is that in the case of unit costs the integrality gap of the {\CLP} is less than $2$,
which resolves a long standing open problem.
Our goal here was just to present the simplest verifiable proof for this fact, and 
we believe that our bound $2-2/15$ can be improved by a slightly more complex algorithm and analysis.
As was mentioned, several LP and SDP relaxations, more complex than the {\CLP}, were shown to have
integrality gap less than $2$ for particular cases (e.g., the {\kBLP} with logarithmic costs).
The hope was that this may lead to ratio better than $2$ for the general case. 
Our result suggests that already the simplest {\CLP}, combined with the dual fitting method,
may be the right one to study to achieve this goal.
More complex LP's (e.g., the {\BLP} or the {\sc Odd-Cut LP}) may be used to improve the ratio.
 
\medskip \medskip

\begin{acknowledgements}
I thank Shoni Gilboa, Manor Mendel, Moran Feldman, and Gil Alon for several discussions.
I also thank Chaitanya Swamy for a discussion on the {\BLP} and the {\sc Odd-Cut LP} during FND 2014,
and L\'{a}szl\'{o} V\'{e}gh on further discussions on the {\BLP} during 2015.
\end{acknowledgements}

% \bibliographystyle{spmpsci}
% \bibliography{tapa}

\end{document}